\documentclass[runningheads,a4paper]{llncs}
\usepackage[margin=1.33in]{geometry}
\usepackage{amssymb}
\usepackage{graphicx}
\usepackage{verbatim}
\usepackage{mathrsfs}
\usepackage[ruled,vlined,lined,commentsnumbered,linesnumbered]{algorithm2e}
\usepackage{amsfonts}

\usepackage{amsthm}
\usepackage{ifpdf}
\usepackage{amsmath}
\usepackage[noend]{algpseudocode}
\usepackage{enumerate}
\usepackage{url}
\usepackage{color,soul}
\usepackage{hyperref}
\usepackage{cite}

\urldef{\mailsa}\path|{ibanerje, richards}@cs.gmu.edu|    
\newcommand{\keywords}[1]{\par\addvspace\baselineskip
\noindent\keywordname\enspace\ignorespaces#1}

\begin{document}

\mainmatter  

\title{Routing and Sorting Via Matchings On Graphs}

\titlerunning{Routing and Sorting Via Matchings On Graphs}

%
%
\author{Indranil Banerjee, Dana Richards}
\authorrunning{I  Banerjee, D Richards}

\institute{George Mason University\\ Department Of Computer Science\\ Fairfax Virginia 22030, USA\\
\mailsa}

%
%

\toctitle{Lecture Notes in Computer Science}
\tocauthor{Authors' Instructions}
\maketitle

\begin{abstract}
The paper is divided in to two parts. In the first part we present some new results for the \textit{routing via matching} model introduced by Alon et al\cite{5}. This model can be viewed as a communication scheme on a distributed network. The nodes in the network can communicate via matchings (a step), where a node exchanges data with its partner. Formally, given a connected graph $G$ with vertices labeled from $[1,...,n]$ and a permutation $\pi$ giving the destination of pebbles on the vertices the problem is to find a minimum step routing scheme. This is denoted as the routing time $rt(G,\pi)$ of $G$ given $\pi$.  In this paper we present the following new results, which answer one of the open problems posed in \cite{5}: 1) Determining whether $rt(G,\pi)$ is $\le 2$ can be done in $O(n^{2.5})$ deterministic time for any arbitrary connected graph $G$. 2) Determining whether $rt(G,\pi)$ is $\le k$ for any $k \ge 3$ is NP-Complete. In the second part we study a related property of graphs, which measures how easy it is to design sorting networks using only the edges of a given graph. Informally, \textit{sorting number} of a graph is the minimum depth sorting network that only uses edges of the graph. Many of the classical results on sorting networks can be represented in this framework. We show that a tree with maximum degree $\Delta$ can accommodate a $O(\min(n\Delta^2,n^2))$ depth sorting network. Additionally, we give two instance of trees for which this bound is tight.  

\end{abstract}

\section{Permutation Routing via Matchings}
Originally introduced by Alon and others \cite{5} this problems explores permutation routing on graphs where routing is achieve through a  series of matchings called steps. Let $G$ be an undirected labeled graph with vertex labeled $i$ having a pebble labeled $p_i$ initially. A permutation $\pi$ gives the destinations of each pebble. The task is to route each pebble to their destination via a sequence of matchings. Given a matching we swap the pebbles on pairs of matched vertices. The \textit{routing time} $rt(G,\pi)$ is defined as the minimum number of steps necessary to route all the pebbles for a given permutation. For given graph $G$, the maximum routing time over all permutations is called  the routing number $rt(G)$ of $G$. Since their inception, permutation routing via matching have generated continual interest. However, prevailing literature focuses on determining the routing numbers of special graphs. We shall give a very brief survey about them in the next section. In this paper we shall focus on the computational aspect of the problem. In particular  we show that for a general graph determining whether $rt(G,\pi)$ is $\le k$ is $\mathsf{NP}$ complete for $k \ge 3$. However, we show that it is possible to determine if $rt(G,\pi) \le 2$ in polynomial time by determining whether a certain graph has a perfect matching. It remains open whether computing $rt(G,\pi)$ is constant factor $\mathsf{APX}$-$\mathsf{Hard}$.

\subsection{Introduction And Prior Results}
The routing via matching model has several variants and generalizations \cite{5,6,8}. For example a popular network routing model is the direct path routing model. In this model a packet move towards its destination directly and no two packets uses the same links (edges). In one version of the problem a path may be specified for each vertex. Costas and others \cite{8} show that this problem and some variants of it to be $\mathsf{NP}$ complete.  In this paper we only consider the classical model as described in the previous section.  In the introductory paper Alon and others \cite{5}  show that for any connected graph $G$, $rt(G) \le 3n$. This was shown by considering a spanning tree of $G$ and using only the edges of the spanner to route the permutation in $G$. Note that one can always route a permutation on a tree, by iteratively moving a pebble that belong to some leaf node and ignoring the node afterwards. The routing scheme is recursive and uses a well known property of trees: a tree has a vertex whose removal results in a forest of trees with size at most $n/2$.
Later in \cite{7} Zhang improve this upper bound $3n/2 + O(\log n)$. This was done using a new decomposition called the caterpillar decomposition. This bound is essentially tight as it takes $\lfloor{3(n-1)/2}\rfloor$ steps to route a permutation on a star $K_{1,n-1}$. There are  few known results for routing numbers of graphs besides trees. We know that for the complete graph and the complete bipartite graph the routing number is 2 and 4 respectively\cite{5}. Later Li and others \cite{6} extends these results to show $rt(K_{s,t}) = \lfloor 3s/2t \rfloor + O(1)$ ($s \ge t$). For the $n$-cube $Q_n$ we know that $n+1 \le rt(Q_n) \le 2n-2$. The lower bound is quite straightforward. The upper bound was discovered using the results for determinig the routing number of the Cartesian product of two graphs \cite{5}. If $G = G_1  \times  G_2$ be the Cartesian product of $G_1$ and $G_2$ then: $$rt(G) \le 2 \min(rt(G_1),rt(G_2))+\max(rt(G_1),rt(G_2))$$
Since $Q_n = K_2 \times Q_{n-1}$ the result follow\footnote{The base case, which computes $rt(Q_3)$ was determined to be 4 via a computer search\cite{6}}. 

Here we take a detour to discuss a related problem of determining the \textit{acquaintance time} of a connected graph. Given a connected graph $G$ whose vertices contains pebbles, its acquaintance time $ac(G)$ is defined to be the minimum number of matching necessary for each pebble to be acquainted with  each other. We say two pebbles are acquainted if they happen to be on adjacent vertices. Hence the acquaintance time of a complete graph is 0. This notion  of acquaintance  was introduce by Benjamini and others in a recent paper\cite{9}. They show that routing number and acquaintance time of a graph are distinct parameters by giving a separation result for the complete bipartite graph. They show $ac(K_{n,n}) = \log n$, which stands in  contrast to the routing number of 4 for $K_{n,n}$. We believe that further investigation is necessary to study graphs which have large separation between the two parameters.

\subsection{Determining whether $rt(G,\pi) \le 2$ Is Easy}
In this section we present a polynomial time deterministic algorithm to determine given a graph if a permutation can be routed in less than two steps. Determining whether $rt(G,\pi) = 1$ is trivial hence we consider only the case when $rt(G,\pi) > 1$.  The basic idea centers around whether we can route the individual cycles of the permutation within 2 steps. Let $\pi = C_1C_2\ldots C_k$ be a permutation with $k$ cycles and $C_i = c_{i_1}\ldots c_{i_j}$. We say a cycle $C$ is  \textit{individually routable } if it can be routed using only edges of the induced subgraph $G[C]$. A pair of cycles $C_1,C_2$ are \textit{mutually routable} if all the pebbles withing them can be routed using only the edges between the two subsets $C_1$ and $C_2$. The next lemma shows that we cannot route  a pair of cycles using edges between the components as well as within the components.

\begin{lemma}
	If a pair of cycles are mutually routable in 2-steps then they must be of same length.
\end{lemma}

\begin{figure}[h]
	\includegraphics[width=4cm]{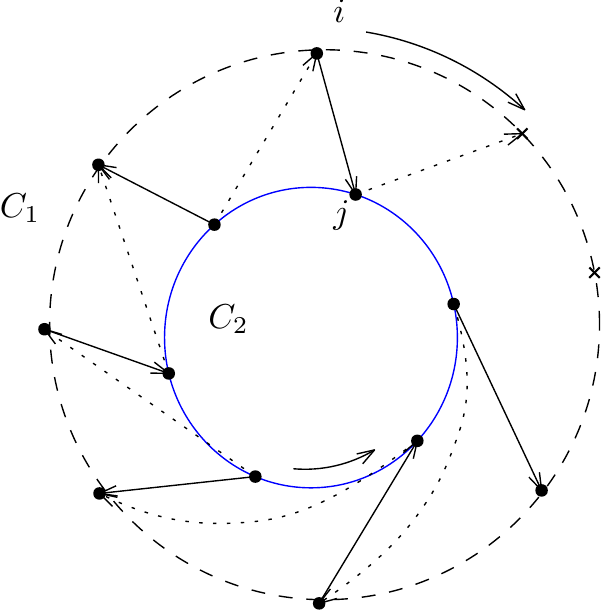}
	\centering
	\caption{The two permutations are shown as concentric circles. The direction of rotation for the outer circle is clockwise and the inner circle is anti-clockwise. Once, we choose $(i,j)$ as the first matched pair, the rest of the matching is forced for both the stages. The crossed vertices in the figure will not be routed.} 
\end{figure}

\begin{proof}
	We prove this assuming $G$ is a complete graph. Since for any other case the induced subgraph $G[C_1\cup C_2]$ would have fewer of edges, hence this is a stronger claim. Let $|C_1| = a \neq  b = |C_2|$. Consider the cycle $C_1=(c_{1_1},\ldots,c_{1_i},\ldots,c_{1_a})$. At the first step we have only three choices for matching some vertex with $c_{1_i}$. 
	\begin{description}
		\item[case 1] \textit{$c_{1_i}$ is not matched}. In this case, in the next(last) round $c_{1_i}$ must be matched with $c_{1_{i+1}}$. This implies $c_{1_{i-1}}$ must be at $c_{1_{i+1}}$ after the first round. This would force $c_{1_{i-2}}$ to be matched with $c_{1_{i+2}}$ in the first round, otherwise $c_{1_{i-2}}$ will not be able to reach $c_{1_{i-1}}$ in two rounds. Proceeding in this way we see that the matching for all the vertices are fixed once we decide not to match $c_{1_i}$. This implies $C_1$ and $C_2$ cannot be routed mutually if we choose to omit any vertex from the matching in $C_1$ in the first round.
		\item[case 2] \textit{$c_{1_i}$ is matched with $c_{1_j}$}. In  this case also we can show that the entire matching is forced. 
		\item[case 3] \textit{$c_{1_i}$ is matched with $c_{2_j}$}. From Figure 1 we see that unless $a  = b$,  the pair $C_1$ and $C_2$ are not mutually routable in 2 steps.\qed
	\end{description}
\end{proof}

\begin{corollary}
	Three or more cycles are not mutually routable in 2 steps. 
\end{corollary}

\noindent Naively verifying whether a cycle $C_i$ are individually or a pair $(C_i,C_j)$ is mutually routable takes $O(|C_i|^2)$ and $O((|C_i| + |C_j|)^2)$ respectively. However, we can employ a more sophisticated approach that takes time proportional to the number of edges in the induced subgraphs corresponding to the cycles. First consider the case of verifying whether a cycle $C=(1,\ldots,i,\ldots,j,\ldots,|C|)$ is individually routable. If $G[C]$ is clique then there are $|C|+1$ different routing schemes, one for each choices of how we match the $i^{th}$ vertex in the first round. For example, if we match $(i,j)$ ($i < j$) in the first then the routing scheme is the following sequence  $(S_1,S_2)$ of matchings (slightly modified from \cite{5}): 

\begin{align}
S_1 = (i,j)(i+1,j-1)\ldots(i-1,j+1)\ldots(i - r \mod n, j + r \mod n)\\
S_2 = (i+1,j)(i+2,j-1)\ldots(i,j+1)\ldots(i - r + 1 \mod n, j + r \mod n)
\end{align} 

\noindent Where $r = \lfloor (n-j+i-1)/2 \rfloor$. Note for every $i$, $S_1(S_2(i)) = i+1 \mod n$. Also, it follows from Lemma 1 that no edge can be in two different routings. Let us label these $|C|$ different routing schemes from 1 to $|C|$. Next we scan through the edges of $G[C]$ maintaining a array of counters of size $|C|$ whose $i^{th}$ element counts the number edges we have seen so far that belongs to the $i^{th}$ routing scheme. After iterating over all the edges if some counter $i$ has a value $|C|$ or $|C|-1$ depending on whether $|C|$ is even or odd, respectively, then we know $C$ can be routed using the $i^{th}$ routing scheme otherwise $C$ is not individually routable. Clearly this takes time linear in the number of edges  in $G[C]$. So the total time to verify all cycles is $O(m)$. We can extend this argument to show that we can also verify all pairs of cycles for mutual routability in $O(m)$ steps also. The main observation is that the pairs of cycles partition the routable edges into disjoints sets. 

Next define a graph $G_{cycle(\pi)} = (V ,E)$ whose vertices are the cycles ($V = \{C_i\}$) and two cycle are adjacent iff they are mutually routable in 2-steps. Additionally, $G_{cycle(\pi)}$ has loops corresponding to vertices which are individually routable cycles. We can modify any existing maximum matching algorithm to check whether $G_{cycle(\pi)}$ has a perfect matching (accounting for the self loops) with only a linear overhead. We omit the details. Then the next lemma follows immediately:
\begin{lemma}
	$rt(G,\pi) = 2$ iff there is a perfect matching in $G_{cycle(\pi)}$.
\end{lemma}
\noindent It is apparent from the previous discussion that $G_{cycle(\pi)}$ can be constructed in $O(m)$ time. Since we have at most $k$ cycles,  $G_{cycle(\pi)}$ will have t most $2k$ vertices and at most $O(k^2)$ edges. 
Hence we can determine a maximum matching in $G_{cycle(\pi)}$ in  $O(k^{2.5})$ time \cite{10}. This gives a total runtime of $O(m + k^{2.5})$ for our algorithm which finds a 2-step routing scheme of a connected graph if one exists.
\begin{corollary}
	$rt(G) = 2$ iff $G$ is a clique.
\end{corollary}
\begin{proof}
	\begin{description}
		\item[$\Rightarrow$ ] A two step routing scheme for $K_n$ was given in \cite{5}.
		\item[$\Leftarrow $] If $G$ is not a clique then there is at least a pair of non-adjacent vertices. Let $(i,j)$ be a non-edge. Then by Lemma 1 the permutation $(ij)(1)(2)\ldots(n)$ cannot be routed in two steps.\qed
	\end{description}
\end{proof}

\subsection{Determining $rt(G,\pi) \le k$ is hard for any $k \ge 3$}
\begin{theorem}
	For $k \ge 3$ computing $rt(G,\pi)$ is $NP$-complete.
\end{theorem}

\begin{proof}
	\begin{figure}[h]
		\includegraphics[width=10cm]{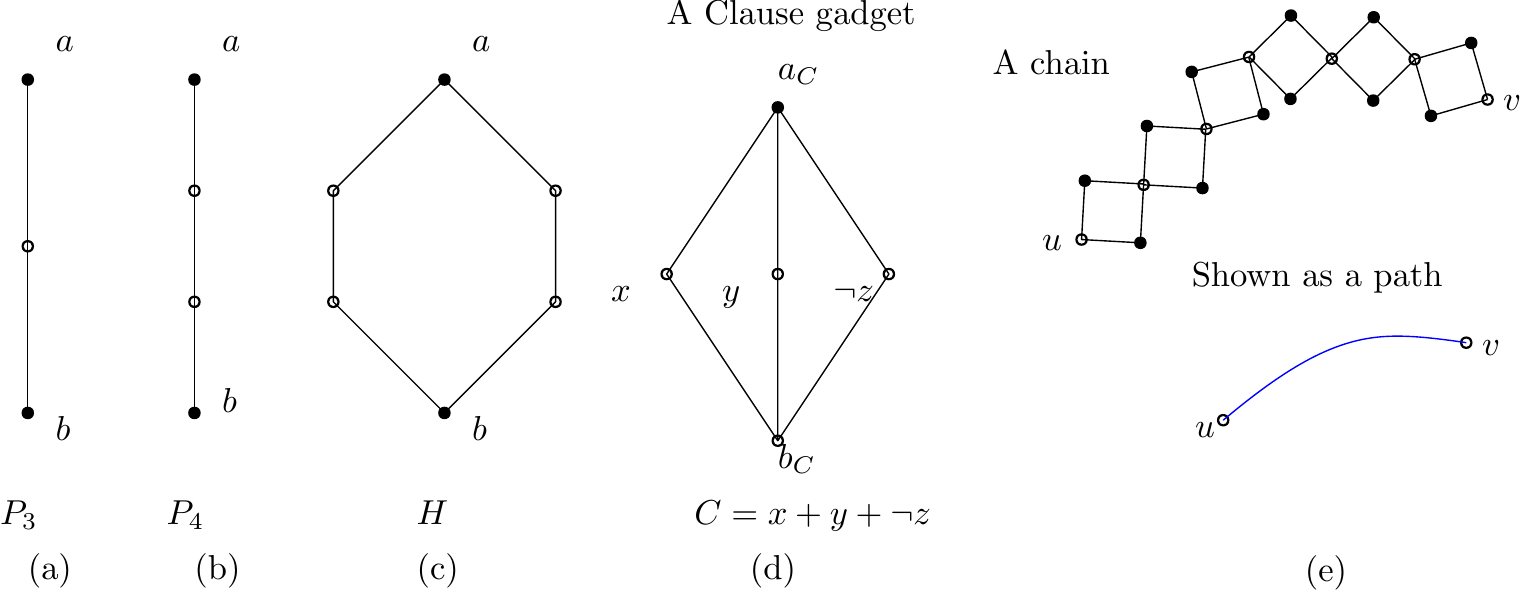}
		\centering
		\caption{Atomic Gadgets, pairs $ (a,b) $ need to swap their pebbles. The unmarked circles have pebbles that is fixed.} 
	\end{figure}
	Proving it is in $NP$ is trivial, we can use a set of matchings as a witness. For the NP hardness proof we first define three \textit{atomic} gadgets (see Figure 6) which will be use to construct the variable and clause gadgets. Vertices whose pebbles are fixed (1 cycles) are represented as circles.  Otherwise they are represented as black discs. So in the first three sub-figures ((a)-(c)) the input permutation is $(a,b)$\footnote{We do not write the 1 cycles explicitly for notational clarity.}. In all our construction we shall use permutation consisting of only 1 or 2 cycles. Each cycle labeled $i$ will be represented as a pair $(a_i,b_i)$. If the correspondence between a pair is clear from the figure then we shall omit the subscript.  It is an easy observation that $rt(P_{3},((a,b))) = rt(P_{4},((a,b))) = rt(H,((a,b))) = 3$. In the case of the hexagon $H$ we see that in order to route the pebbles within 3 steps we have to use the left or the right path, but we cannot use both paths simultaneously (i.e., $a$ goes through left but $b$ goes through the right). Figure 6(e) shows a chain of squares connecting $u$ to $v$. If vertex $u$ is used during routing any pebble other than the two pebbles to its right then the chain construction forces $v$ to be used in routing the two pebbles to its left. This chain is called a \textit{f-chain}. In our construction we  use chains of constant length to simplify the presentation of our construction.  
	\begin{figure}[h]
		\includegraphics[width=12cm]{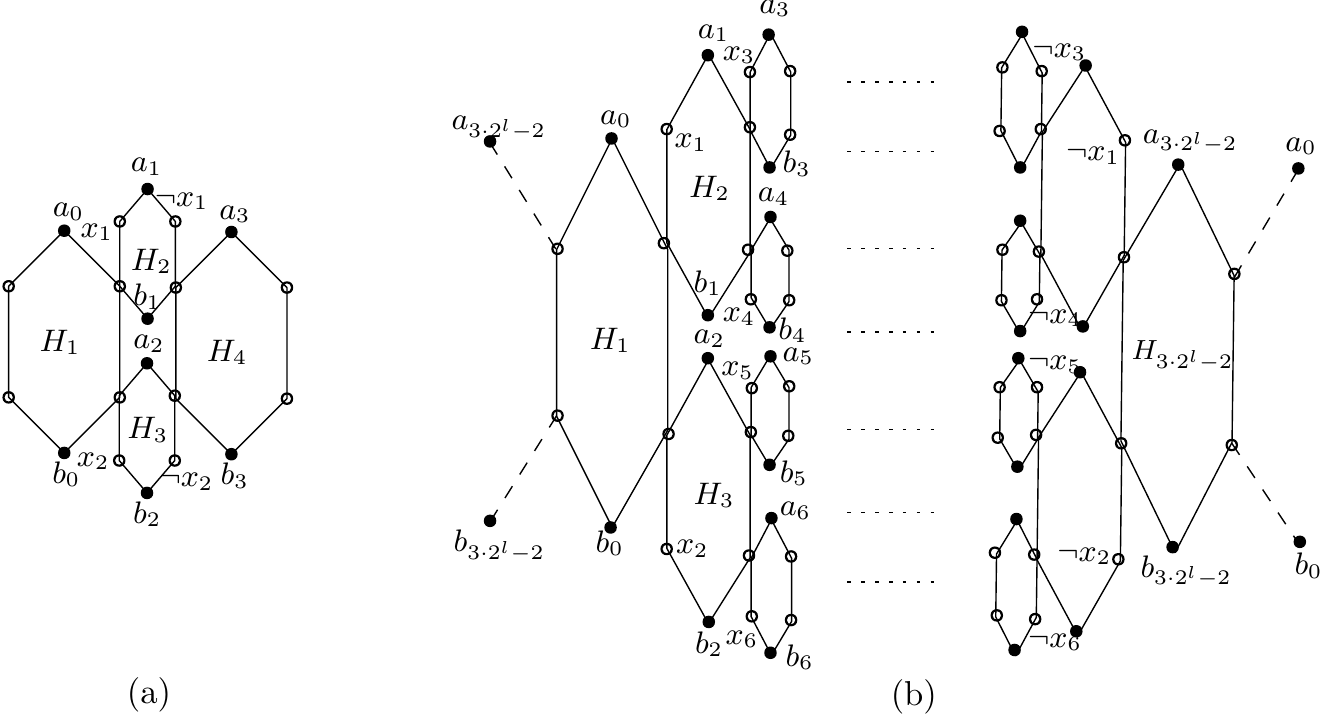}
		\centering
		\caption{Variable graph of $X$. (a) is a special case for $l=1$, (b) is the general case.} 
	\end{figure}
	\subsubsection{Clause Gadget:} Say we have clause $C = x \vee y \vee \neg z$. In Figure 6(d) we show how to create a clause gadget. This is referred to as the \textit{clause graph} $G_C$ of the clause $C$. The graph in Figure 6(d) can route $\pi_C=(a_c,b_C)$ in three steps by using one of the three paths between $a_C$ and $b_C$. Say, $a_C$ is routed to $b_C$ via $x$. Then it must be the case that vertex $x$ is not used to route any other pebbles. We say the vertex $x$ is \textit{owned} by the clause. Otherwise, it would be not possible to route $a_C$ to $b_C$ in three steps via $x$. We can interpret this as follows. A clause has a satisfying assignment iff its clause graph has a owned vertex.
	
	\subsubsection{Variable Gadget:} Construction of the variable gadgets are little more involved than the clause gadgets. For some $l > 0$, let the variable $X$ is in $m_X \le 2^{l+1}-2$ clauses. The variable gadget corresponding to $X$ is shown in Figure 7(b). Vertically aligned hexagons are all in one level. Number of levels is $2l+1$. The left most hexagon $H_1$ and the rightmost hexagon $H_{3\cdot 2^{l}-2}$ share a common edge as indicate in the figure making it circularly wrapped. The permutation we will route on $G_X$ (the variable graph of $X$) is $\pi_X = (a_0b_0)(a_1b_1)(a_2b_2)\ldots(a_{2^{l+1}-2}b_{2^{l+1}-2})$. For each variable we shall have a separate graph and a corresponding permutation on its vertices. In the graph $G_X$ there are only two possible ways to route $\pi_X$ in two steps. 1) If we route $(a_0b_0)$ using the right path in $H_1$ this forces $(a_1b_1)$ and $(a_2b_2)$ to be routed using the right paths in their respective hexagons $H_2$ and $H_3$. Continuing in this way we see that $(a_{2^{l+1}-2}b_{2^{l+1}-2})$ must be routed using the right path of $H_{3\cdot 2^{l}-2}$. In during this routing the vertices $x_1,x_2,\ldots,x_{2^{l+1}-2}$ are not used and hence are free and can be owned by some clause. 2) If we route $(a_0b_0)$ using the left path of $H_1$, the opposite happens and $\neg x_1,\neg x_2,\ldots,\neg x_{2^{l+1}-2}$ will be the free vertices in this case. This forces variable assignment. The former and latter case corresponds to true (right) and false (left) assignment of $X$ respectively.
	
	\subsubsection{Reduction:} For each clause $C$, if the literal $x \in C$  then we connect $x_i \in G_X$  (for some $i$) to the vertex labeled $x \in G_C$ via an \textit{f-chain}. If $\neg x \in C$ then we connect it with $\neg x_i$ via an \textit{f-chain}. This is our final graph $G_{\phi}$ corresponding to an instance of a 3-$ \mathtt{SAT} $ formula. The input permutation is $\pi = \pi_X\ldots\pi_C\ldots\pi_f\ldots$, which is the concatenation of all the individual permutation on the variable, clause  graphs and $f$-\textit{chains}.  This completes our construction. We need to show, $rt(G_\phi,\pi) = 3$ iff $\phi$ is satisfiable. Suppose $\phi$ is satisfiable. Then for each variable $X$, if the literal $x$ is true then we use left routing in $G_X$, otherwise we use right routing. This ensures in each clause graph there will be at least one owned vertex.  Now suppose $(G_{\phi},\pi) = 3$. Then each clause graph has at least one owned vertex. If $x$ is a free vertex in some clause graph then $\neg x$ is not a free vertex in any of the other clause graphs, otherwise variable graph $G_X$ will not be able rout its own permutation in 3 steps. Hence the set of free vertices will be a satisfying assignment for $\phi$. It is an easy observation that the number of vertices in $G_\phi$ is polynomially bounded in $n,m$; the number of variables and clauses in $\phi$ respectively and that $G_\phi$ can be explicitly constructed in polynomial time.

\end{proof}

\section{Optimal Sorting Network For A Given Graph}
In this section we introduce {\it sorting numbers} for graphs in the context of  sorting networks. Majority of existing literature on sorting networks focus on optimizing the depth (number of concurrent stages) and size (total number of comparators used) of such networks. Here, we study a slightly different problem. Let the \textit{sorted order} of the vertices of a labeled graph $G$ be the permutation $\pi$ that assigns a rank $\pi(i)$ to the vertex labeled $i$.
Given $G$ the task is to design an optimal sorting network (in terms of depth) on $G$. Before we make this notion formal we need to define a sorting network.

\begin{definition}[Sorting Network]
	A sorting network is a triple $\mathcal{S}(H,M,\pi)$ such that:
	\begin{enumerate}
		\item $H$ is a connected labeled graph having $n$ vertices and a sorted ordering $\pi$ on its vertices. Initially each vertices of $H$ contains a pebble having some value, that is they act as input terminals of the network.
		\item The ordered set $M$ consists of directed matchings in $H$. In a directed matching some edges in the matching have been assigned a direction. Sorting occurs in stages.
		At stage $i$ we use the matching $m_i \in M$ to exchange pebbles between matched vertices according to their orientation. For an edge $\overrightarrow{uv}$, when swapped the smaller of the two pebble goes to $u$. If an edge is undirected then both pebbles swap regardless of their order.
		\item After $|M|$ stages the vertex labeled $i$ contains the pebble whose rank is $\pi(i)$ in the sorted order of the pebbles.  $|M|$ is called the depth of the network. Additionally, this must hold for all ($n!$) initial arrangement of the pebbles.
		\item Each edge in $H$ is in some matching, that is $H$ is minimal. 
	\end{enumerate}
	
\end{definition}

\noindent We say  $\mathcal{S}(H,M,\pi)$ is a sorting network on $G$ if $H$ is a spanning subgraph of $G$. Let $\mathcal{S}_G$ be the set of all such sorting networks on $G$ over all possible spanning subgraphs and sorted orderings (non-isomorphic).

\begin{definition}[Sorting Number]
	Sorting number $st(G)$ of a graph $G$ is defined to be minimum depth of any sorting network on $G$. Additionally, $st(G,\pi)$ is the sorting number of $G$ over all possible sorting network on $G$ with a fixed sorted order $\pi$.
\end{definition}

\begin{lemma}
For any $\pi$ we have $st(G, \pi) \le st(G) + O(n)$.
\end{lemma}

\noindent The above lemma implies that if we construct a sorting network for some arbitrary sorted order on the vertices then we suffer a penalty of $O(n)$ on the depth of our network as compared to the optimal one. 
\begin{proof}
	Let $\mathcal{S}(H, M, \pi^*)$ be an optimal sorting network on $G$. Using this network we can create another sorting network $\mathcal{S}(H, M, \pi)$ whose depth is $O(n)$ more than the optimal one. This can be done in two rounds. First we use the sorting network $\mathcal{S}(H, M, \pi^*)$ to determine after $st(G)$ stages the ranks of each pebble in the sorted order. After this step we will know that the pebble at vertex $i$ has a rank $\pi^*(i)$. But $\mathcal{S}(H, M, \pi)$ sends a pebble with rank $\pi^*(i)$ to the vertex labeled $\pi^{-1}(\pi^*(i))$. Hence we can route the fixed permutation $\pi^{-1}\pi^*$ on $G$ during the second round to arrive at the desired sorted order $\pi$. Since we can route a permutation on any graph in $O(n)$ steps the lemma follows.\qed

\end{proof}

\noindent Note that if $G$ is not connected than $st(G) = \infty$. Otherwise, there always exists a spanning tree $T$ of $G$ and $st(G) \le st(T)$. The main result of this section will be to obtain both a lower and an upper bounds for $st(T)$. We start by restating some previous results for sorting networks with restricted topology under this new framework. The path graph $P_n$ is one of the simplest case. We know that $st(P_n) = 2n$. This follows from the fact that the classical odd-even transposition sort takes $2n$ matching steps and that is optimal. Next we discuss some known bounds for the sorting numbers of some common graphs starting with the complete graph. These results are summarized in Table 1. For the the complete graph $K_n$. Ajtai-Komlos-Szemerdi (AKS) sorting network directly gives an upper bound of $O(\log {n})$ for the sorting number of $K_n$. In this case also the bound is tight. For the  $n$-cube $Q_n$ we can use the Batcher's Bitonic sorting network, which has a depth of $O((\log n)^2)$ \cite{2}. This was later improved to $2^{O(\sqrt{\log\log n})}\log n$ by Plaxton and Suel\cite{11}.  We also have a lower bound of $ \Omega(\frac{\log n \log\log n}{\log \log \log n})$ due to Leighton and Plaxton. For the square mesh $P_n \times P_n$ it is known that $st(P_n \times P_n)  = 3n + o(n)$, which is tight with respect to the constant factor of the largest term. This follows from results of Schnorr \& Shamir \cite{4}, where they introduced the $3n$-sorter for the square mesh.

\begin{table}[ht]
	\caption{Known Bounds On The Sorting Numbers Of Various Graphs} 
	\centering 
	\begin{tabular}{c c c c} 
		\hline\hline 
		Graph & Lower Bound & Upper Bound & Remark \\ [0.5ex] 
		\hline 
		Complete Graph ($K_n$) & $\log n$ & $O(\log n)$ & AKS Network \cite{1}\\ 
		Hypercube ($Q_n$) & $ \Omega(\frac{\log n \log\log n}{\log \log \log n}) $\cite{11} &  $2^{O(\sqrt{\log\log n})}\log n$ & Leighton and Plaxton \cite{13} \\
		Path ($P_n$) & $n-1$ & $2n$ & Odd-Even Transposition Sort \\
		Mesh ($P_n \times P_n$) & $3n - 2{\sqrt{n}}-3$ & $3n+O(n^{3/4})$ & Schnorr \& Shamir \cite{4}\\ 
		Tree & $ \Omega(n^2)$ & $ O(\min(\Delta^2n,n^2)) $& this paper\\[1ex] 
		\hline 
	\end{tabular}
	\label{table:nonlin} 
\end{table}

\begin{figure}[h]
	\includegraphics[width=6.5cm]{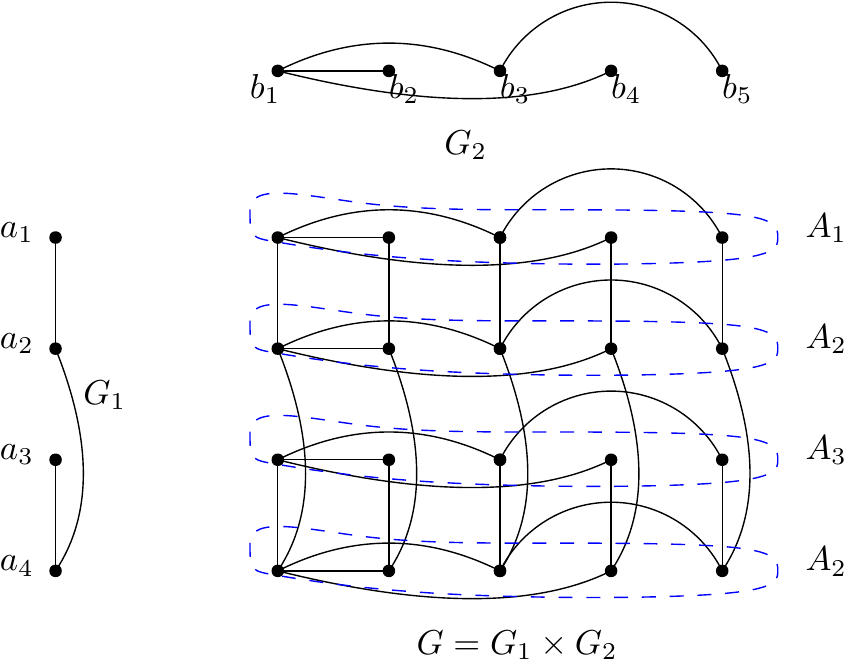}
	\centering
	\caption{The product graph $G = G_1\times G_2$. The rows highlighted by blue regions represents the super vertices of $G'$. } 
\end{figure}

\noindent Given two graphs $G_1(V_1,E_1)$ and $G_2(V_2,E_2)$ their Cartesian product is the graph $G(V,E)$ such that $V = V_1 \times V_2$ and $E = \{((u_1,u_2),(v_1,v_2))| u_1,v_1 \in V_1, u_2,v_2 \in V_2\ \mbox{and either}\ u_1=v_1\ \mbox{or}\ u_2=v_2  \}$. The next theorem gives the sorting number of a product graph in terms of sorting numbers of its components.
\begin{theorem}
	Given two graphs $G_1$ and $G_2$, the sorting number of their Cartesian product is $$st(G_1 \times G_2) \le st(G_1)st(G_2) + st(G_1)+ st(G_2)$$
\end{theorem}

\noindent The network that sorts $G = G_1 \times G_2$ is constructed via networks for $G_1$ and $G_2$ as follows. Let, $V_1 = \{a_1,...,a_{n_1}\}$ and $V_2=\{b_1,...,b_{n_2}\}$ be the vertex set of the graphs $G_1$ and $G_2$ respectively. The labeling of the vertices are based on the output ranks of the sorted order. The graph $G$ with vertex set $V=\{(a_i,b_j)\}$ can be visualized in a grid of size $n_1n_2$. Each row consists of a copy of $G_2$ and each column consists of a copy of $G_1$. See Figure 2 for an illustration. 

The sorting network for $G$ consists of the following matching scheme. Let 
$M_1$ and $M_2$ be the respective matching schemes for $G_1$ and $G_2$. If $(u,v)$ is a pair of matched vertices in some matching, we assign a direction to the edge $uv$  according to the comparator attached to the edge $uv$. If smaller of the two pebbles is put in $u$ after the exchange then we say the edge is directed from $u$ to $v$. Then $u$ is called the lower vertex and $v$ is called the upper vertex. A directed matching thus partitions the vertex set into three parts: upper, lower and non-participating vertices. Let, $M_1 = (m_{1},...,m_{{st(G_1)}})$ and $M_2 =  (m'_{1},...,m'_{{st(G_2)}})$.

We start by sorting each row of $G$, which have copies of $G_2$ using the sorting network $M_2$. However, each row corresponds to a vertex in $G_1$.  Consider the set of upper, lower and non-participating vertices of $G_1$ for the matching $m_1$. These vertices partition the rows of $G$ into three parts. For each row in $G$ if it is associated with a lower vertex in $G_1$ then we call it a lower row. Similarly we define upper rows and non-participating rows. For each lower row then we sort it normally using the sorting network of $G_2$. If the row is an upper row we sort it using the sorting network of $G_2$ where the direction of the comparators have been reversed. We leave the non-participating rows unmatched. Next we use the matching $m_1$ to do a compare exchange on the columns of $G$. These two stages (sorting on rows (copies of $G_2$) and the application of a matching from $M_1$) together constitute a single \textit{full} stage in $G$. The set of matchings without the final compare-exchange on columns constitute a half-stage. Hence a full stage consists of $st(G_2) + 1$ matchings on $G$. Continuing, we invoke a full stage corresponding each successive matching in $M_1$, hence for $st(G_1)$ full stages. At the end we need to sort every row of $G$ in ascending order.  This last stage is a half stage and adds an additional $st(G_2)$ matchings  to the sorting network. The final sorted order of vertices are $((a_1,b_1) \le (a_1,b_2) ... \le (a_1,b_{n_2}) \le (a_2,b_1) ... (a_2,b_{n_2})\le...\le(a_{n_1},b_1)...\le (a_{n_1},b_{n_2}))$.

\begin{proof}[proof of correctness]	 
	The correctness of the above procedure can be proven using the 0-1 principle. Each half stage in $G$ consists of sorting in ascending order or descending order. This is followed up by a compare exchange between the matched rows. Consider a pair of matched rows $A_i = ((a_i,b_1),...,(a_i,b_{n_2})$ and $A_j = ((a_j,b_1),...,(a_j,b_{n_2})$ corresponding to vertices  $a_i$ and $a_j$ in $G_1$. Assume $i$ precedes $j$ in the sorted order in $G_1$ ($i < j$). Since $A_i$ is sorted in ascending order and $A_j$ in descending order, a compare-exchange between the pairs $(((a_i,b_1),(a_j,b_1)),...,((a_i,b_{n_2}),(a_j,b_{n_2})))$ is a merge operation for 0-1 input. Hence, after the compare exchange we have $u \le v$ for every pair of vertices, where $u \in A_i$ and $v \in A_j$. It is known that \cite{12} if we replace every comparator of a sorting network by a merging subnetwork then the new network correctly sorts every input sequence whose elements are multi-sets  instead of single elements. In our case $A_i$'s corresponds to these multi-sets and the sorting operation (on copies of $G_2$) followed by exchange correspond to a single merging operation in the network $G_1$. \qed 
\end{proof}
Recall the analogous result for the routing number of the product graph \cite{5}. We have $rt(G_1 \times G_2)\le 2 \min(rt(G_1), rt(G_2)) + \max(rt(G_1), rt(G_2))$. The corresponding bound for the sorting number is much worse. Since a $n$-cube $Q_n$ can be written as the Cartesian product of $Q_{n-1} \times K_2$, from Theorem 1 we see that $d(Q_n) \le O(3^n)$, which is $O(N^{\log{3}})$ where $N = 2^n$ is the number of vertices in a $n$-cube. Unfortunately, although non-trivial, the above bound is weak for $n$-cubes. However, we believe that for $K_{1,n-1} \times K_{1,n-1}$ (product of two stars) the bound of the theorem may be tight. 

\subsection{Sorting Number Of Trees}
Here we informally discuss the lower bound of $st(T)$. This occurs when the tree is a star. For a star $K_{1,n-1}$ there are only $n$ non-isomorphic sorted orders. Without loss of generality we assume in the ordering the center gets the pebble ranked $n$. Let $M = (e_1,\ldots,e_{st(T)})$ be a sequence of matchings, which are in this case are just singleton edges. The important observation is this: once a pebble is placed in its final sorted position, it must stay there for the remainder of the matchings for $st(T)$ to be minimum. Given $M$, consider the input permutation of pebbles (given by an adversary) which makes the first pebble  to be put into its correct place be first matched at least after $(n-1)$ steps. Similarly, the second pebble to be matched at least after an additional $n-2$ steps and so on. This would ensure that it takes at least $\Omega(n^2)$ steps to obtain the sorted order. Note, if the sorted order had put some other pebble ($i$) than the pebble ranked $n$ at the center, then it takes at most 1 additional step to put $n$ at the center.
\subsection{An Upper Bound}
	\begin{figure}[h]
		\includegraphics[width=4.5cm]{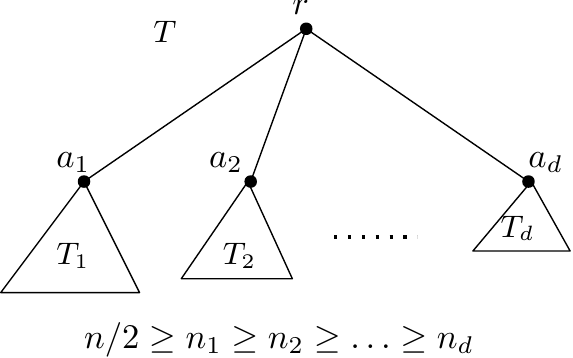}
		\centering
		\caption{A balanced decomposition of a tree. } 
	\end{figure}
	
		\begin{algorithm}[H]
			\SetKwData{Left}{left}\SetKwData{This}{this}\SetKwData{Up}{up}
			\SetKwFunction{Sort}{Sort}\SetKwFunction{fc}{MaximalLayer}
			\SetKwInOut{Input}{Input}\SetKwInOut{Output}{Output}
			\SetKwFor{Pardo}{pardo}{do}{endfor}
			\Input{$T$ with root $r$}
			\Output{Pebbles are sorted according to an MP labeling}

			\Begin{
				\If{$|T| == 1$}{\KwRet}\tcp*[l]{Begining of phase-1}
				\For{$i$ from $d-1$ to 1}{
					\For{$j$ from 1 to $i$}{
						$\mathsf{Swap}(T_j,T_{j+1};r)$;
					}
				}
				\tcp*[l]{Begining of phase-2}
				\Pardo{$i$ from 1 to $d$}{
					\If{$i$ == 1}{$\mathsf{OddEvenTreeSort}(T_1',r_1)$}\Else{$\mathsf{OddEvenTreeSort}(T_i,r_i)$}
				}
				
			}
			\caption{$\mathsf{OddEvenTreeSort}(T,r)$}
		\end{algorithm}

In this section we present an oblivious sorting algorithm for trees. The algorithm  $\mathsf{OddEvenTreeSort}$ is a natural generalization of the classical odd-even transposition sort algorithm. First recall the following fact about trees: for any tree $T$ with $n$ vertices there exists a vertex $r$ whose removal produces connected components of size $\le n/2$.  We can take this special vertex $r$ as the root, which we assume has $d$ children, see Figure 3. Let the subtree $T_i$ have $n_i \le n/2$ nodes and any non-leaf node has at most $\alpha_i$ children. Further, assume the subtrees are arranged in descending order according to their size from left to right ($n_1 \ge n_2 \ge \ldots \ge n_{d}$). The sorted order $\pi(T)$ is defined recursively as follows.  

\textit{\begin{enumerate}
		\item $T_1 \cup \{r\}$ will have the  $n_1 + 1$ smallest pebbles.
		\item The tree $T_i$ has pebbles whose ranks are between $\sum_{j=1}^{i-1}n_j+2$ to $\sum_{j=1}^{i}n_j + 1$. 
		\item Labeling of each subtree $T_i$ is defined recursively based on an appropriately chosen root $r_i$ which partitions $T_i$ in a balanced manner.
	\end{enumerate}}

\noindent We call this the \textit{\underline{m}ulti-root \underline{p}re-order} (MP) labeling. Note that the root $r_i$ of $T_i$ may not be the the $a_i$ in Figure 3. Given a tree this ordering can be easily precomputed and it is fixed afterwards. Furthermore once we have our sorting network, using Lemma 8 we can easily create another sorting network with a more natural ordering\footnote{For example we can use the pre-order ranks of the vertices in $T$ as our sorted order.} using an additional $O(n)$ steps.
 The $\mathsf{OddEvenTreeSort}(T, r)$ has two main phases: 1) In phase-1 we use the subtrees as buckets to partition the pebbles such that $T_1' = T_1\cup\{r\}$ gets the first $n_1+1$ smallest pebbles, $T_2$ the next $n_2$ smallest pebbles and so on. 
2) Next in phase-2 we call $\mathsf{OddEvenTreeSort}(T_i,r_i)$ recursively for all the subtrees $T_1',\ldots,T_{d}$. Sorting on these subtrees happens in parallel. Let as assume the number of matchings needed to partition the pebbles during the first phase is $ S(n,d; \alpha_1,\ldots,\alpha_{d}) $, if the root $r$ has degree $d$ and the subtree $T_i$ has maximum arity of $\alpha_i$. Then the total number of stages in $\mathsf{OddEvenTreeSort}$ is given by the following recurrence:

\begin{align}
C(T) = \max(C(T'_1),\max_i{C(T_i)}) + S(n,\Delta; \alpha_1,\ldots,\alpha_d)
\end{align}


Since any routing between the subtrees must use the root $r$, we route the pebbles between a pair of subtrees at a time. The procedure $\mathsf{Swap}(T_i,T_{i+1};r)$ takes a pair of consecutive subtrees and sorts the pebbles in such a away that after completion each pebble in $T_i' = T_i \cup \{r\}$ is smaller than pebbles in $T_{i+1}$. We first describe in detail $\mathsf{Swap}(T_i,T_{j};r)$. Figure 4 shows the two subtrees $T_i$ and $T_j$ connected via $r$. Let the height of $T_i$ be $h_i$ and the height of $T_j$ be $h_j$. Total number of levels in $T_{ij}$ is $h_i+h_j+1$. Vertices are grouped according to the level they are in starting from the leftmost vertices (which are at depth $h_i$ in $T_i$) which are assigned to group 1. Based on its group number a vertex is either an odd or an even vertex. Each stage consists of either matchings between odd-even vertices or even-odd vertices. For each non-leaf node we pick an arbitrary but fixed ordering of its children so that at any odd or even stage they will be chosen sequentially in that order. This also makes the above scheme oblivious. All the matched edges are directed from left to right. For each pair of matched vertices we exchange their pebbles if the vertex to the left has a larger pebble to that of the right. We call an odd followed by an even stage together a cycle. We shall count the number of cycles to simplify our analysis. 

\begin{lemma}
	Assuming $n_i \ge n_j$, the procedure $\mathsf{Swap}(T_i,T_{j};r)$ requires  at most $2(n_i + \max(\alpha_i,\alpha_j)p_{ij}) + O(1)$ cycles to route the pebbles to their subtrees. Where $p_{ij}$ is the number of pebbles transported from $T'_i$ to $T_j$.
\end{lemma}

\begin{proof}
	\begin{figure}[h]
		\includegraphics[width=5cm]{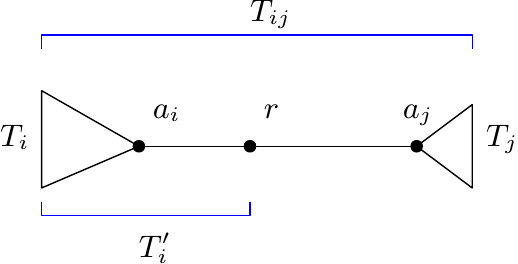}
		\centering
		\caption{Pair of subtrees joined at the root $r$.} 
	\end{figure}
	The procedure $\mathsf{Swap}(T_i,T_{j};r)$ is oblivious, thus by the 0-1 principle \cite{12} we only need to show it works correctly when the input is restricted to 0 and 1. $\mathsf{Swap}(T_i,T_{j};r)$  is broken up into two rounds. In the first round, which happens in parallel on the subtrees $T_i'$ and $T_j$, we move the larger  pebbles towards $r$ in $T_i'$ and we move the smaller pebbles towards $a_j$ in $T_j$. After completion of this round every path from $r$ to a leaf in $T_i'$ is  decreasing and for $T_j$ every path from $a_j$  to a leaf is increasing.
	\begin{figure}[h]
		\includegraphics[width=14cm]{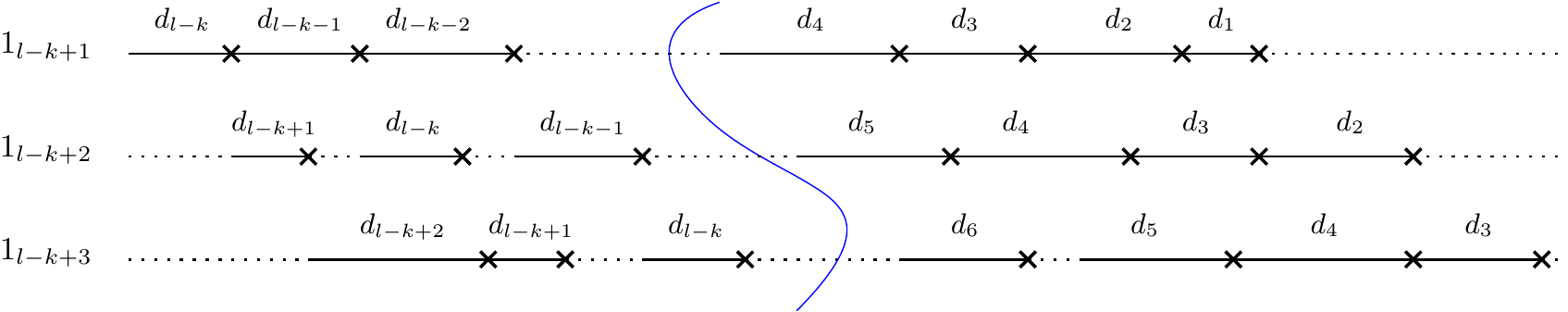}
		\centering
		\caption{Partial timing diagram for the path $P$. } 
	\end{figure}
	
	\begin{claim}
		For a tree $T$ having $n$ vertices, $2n$ cycles are sufficient to achieve this ordering.
	\end{claim}
	To prove this claim consider a path $P$ in $T$ from $r$ to some leaf node $u_l$ of $T'_i$. Let $P = (u_1,\ldots,u_l)$ where $r=u_l$. Let $D=(d_1,\ldots,d_{l-1})$ be the arities of the nodes on this path from $r$ to the leaf. We prove that for this path at most $2\sum_{i=0}^{l-1}{d_i}$ cycles suffice to sort the pebbles on it (that is all the 1's that were initially in the path are closer to the root than the 0's). Here we make two important observations that simplifies our analysis: 1) No new 0's are introduced to this path during any stage of the routing. 2) If a new 1 enters this path from a child not on this path replacing some 0, we can essentially ignore it, as this does not hurt the relative ordering of the existing 0's and 1's in the path. So we assume no new 1 is introduced in the path during the routing. This two assumptions allows us to treat the path $P$ in isolation. Since $P$ was chosen arbitrarily proving the property holds for $P$ would suffice. Now, consider the special case when $d_0=1,\ldots,d_{l-1}=1$. Then the odd-even matchings in $P$ would follow the same pattern as the Odd-Even Transposition sorting network, which takes $l \le 2\sum_{i=0}^{l-1}{d_i}$ cycles. The proof for this special case can be extended relatively easily with some additional bookkeeping. Let as assume that we have $k$ 1's initially in $P$. It is apparent that the worst case happens when all the 1's are initially at the other end of the root. Assuming the root to be the rightmost vertex in $P$. Let the 1 initially at vertex $u_i$ be labeled as $1_i$. The timing diagram in Figure 5 shows the progression of each 1 towards the root. The rightmost 1 is at vertex $u_{l-k+1}$ initially. When counting the number of cycles we make the following conservative estimate. It takes $1_{l-k+1}$ at most $d_{l-k}$ cycles to move to $u_{l-k}$. Similarly it takes  $d_{l-k}+d_{l-k+1}$ cycles for $1_{l-k+2}$ to move to $u_{l-k+1}$. Since it can take $1_{l-k+1}$ at most $d_{l-k}$ to vacate $u_{l-k+1}$ and at most $d_{l-k+1}$ additional cycles for the edge $u_{l-k+1}u_{l-k+2}$ to be in a matching afterwards. In Figure 5, the crosses represent a jump from child to its parent by a 1. The thick lines represent the cycles when the child with a 1 is waiting to be matched with its parent (which has a 0) so that it can move up. The dotted line between two thick line represents the time spent by the child idling whose parent still has a 1. It is clear from the Figure 5 that the maximum bottleneck occurs when $d_{l-1} \le d_{l-2} \ldots \le d_{1}$. Otherwise, if we have $d_i \le d_{i+1}$ for some $i$ then the child at $u_{i+3}$ have to wait $d_{i+1} + d_{i+2}$ to move to $u_{i+1}$. Since, $d_{i} \le d_{i+1}$ by the time this 1 reaches $u_{i+1}$ the 1 at $u_{i}$ would have already moved on to $u_{i-1}$. Which only helps with the routing. Figure 5 shows the timing diagram for the first three 1's starting from the right most 1 at $u_{l-k+1}$. We see that for $1_{l-k+1}$ it takes at most $\sum_{i=1}^{l-k}{d_i}$ cycles to reach the root $u_1$. Similarly for $1_{l-k+2}$ it takes at most $\sum_{i=1}^{l-k+1}{d_i} + \sum_{i=1}^{1}{d_{1+i}}$ cycles to reach $u_1$. In general for $1_{l-k+j}$ it takes at most 
	\begin{align}
	\sum_{i=1}^{l-k+j-1}{d_i} + \sum_{i=1}^{j-1}{d_{1+i}} \le 2\sum_{i=1}^{l-1}{d_i}
	\end{align}
	cycles. Clearly for any path, $\sum_{i=1}^{l-1}{d_i} \le n$ which gives us the first term in  the bound of the lemma. This proves the claim.\qed
	
	In the second round we exchange pebbles between $T_i'$ and $T_j$ such that after end of this round at least $T_i'$ is all 0's or $T_j$ is all 1's. We show this rounds takes at most $2\max(\alpha_i,\alpha_j)p_{ij} + O(1)$ cycles. This can be proven by recycling some of the main ideas from the previous proof. Again we  consider a path, but this time not restricted to the subtrees themselves but from one subtree to another. Let this path be $P = (u_{-l_L},u_{-l_L+1},\ldots,a_i=u_{-1},r=u_0,a_j=u_1,\ldots,u_{l_R})$. Let the arity sequence be analogously defined ($d_{-l_L+1},\ldots,1,1,d_1,\ldots,d_{l_R - 1}$). Note here that, the node $r$ links the two trees and contributes an arity of 1 for each side (Figure 4). Let $P = P_L | P_R$ (where `$|$' represents concatenation), where $P_L= (u_{-l},\ldots,u_0)$ and $P_R = (u_1,\ldots,u_l)$.  We call a $P_L$ with all 0's or a $P_R$ with all 1's as \textit{clean}, otherwise they are \textit{dirty}. After the completion of this round we assume without loss of generality that $T'_i$ is clean (all 0's). Hence, after completion of this round every $P_L$ will be pure. Conversely, as long as their is some impure $P_L$ their must also be some impure $P_R$. Hence if $P_R$ becomes pure before $P_L$ we can consider another path $P'=P_L | P'_R$ where $P'_R$ is still impure. Since the paths $P$ and $P'$ have consumed equal number of cycles before $P_R$ became pure, due to symmetry, we can carry our timing diagram (of pebbles in $P_L$) over to $P'$ and continue extending it. Similarly we can extend our timing diagram to other paths $P'', P''' \ldots$ etc if necessary. Hence, without loss of generality we assume $P_R$ never becomes pure before $P_L$. Also note that after crossing $u_1$, a 1 in this path may choose to move to some other subtree not in the path. Hence, from the perspective of $P$ we can imagine that the said 1 just vanishes and is replaced by a 0. This however is a desirable situation and only helps with our bound since we want $T'_i$ to be clean\footnote{This argument is not symmetric. If $T_j$ was clean after this round then we just do the same analysis from the perspective of the 0's in $T_j$.}.

	Now we look at the routing schedule on $P$. Before beginning of this round $u_i$'s have all 1's starting from $u_0$ up to some $t \ge -l_L$. We can carry over most of the observations we made to obtain the timing diagram for the first round.  However, we must take into account one additional bottleneck. During the analysis of the first round we ignored the 1's coming from other subtrees to our path $P$ as they did not hinder the invariant (all 1's precedes all 0's). However, we cannot ignore these extraneous 1's for this round because all the 1's in $P_L$ must move to $P_R$ before the end of it. At the beginning, the sentinel $1_{0}$ moves to $u_1$ immediately after 1 cycle. Then it may take it up to $d_1$ cycles to move to $u_2$. For $1_{-1}$ it takes at most 2 cycles to move to $u_0$. Actually the timing diagram for this two 1's will be exactly the same as in the first round because no other 1 from some other subtree can get in front of them. Let $l'_L=\min(l_L,p_{ij})$ and $l'_R=\min(l_R,p_{ij})$ . For the sentinel it takes at most  $\sum_{i=1}^{l'_L}d_i+1$ to reach $u_{l'_R}$ or finish at somewhere left of $u_{l'_R}$. Similarly it takes $1_{-1}$ at most $\sum_{i=1}^{l'_L}d_i + 2 + d_{l'_R-1}$ steps to reach its final position.  
	However things get interesting for $1_{-2}$. First it has to wait for its predecessor $1_{-1}$ to vacate $u_{-1}$. Once $1_{-1}$ has left $u_{-1}$ it may be the case that, an arbitrary number of 1's from other sibling subtrees may end up beating it to reach $u_{-1}$. Let us say $s_{2}$  1's reach $u_{-1}$ before $1_{-2}$ does. Then it would take $1_{-2}$ at most $\sum_{i=1}^{l'_L}d_i + 2 + \sum_{i=1}^{s_{2}+1}d_{l'_R-i}$ to reach $u_{l'_R-2-s_2}$ or end up somewhere left of it. We note that the second summation dominates the first one the further left we go. This can be generalized: if $s_{t}$ 1's have moved into $P$ in front of $1_{-t}$ (these events occur before $1_{-t}$ reaches $u_{-1}$) then it takes $1_{-t}$ at most $\sum_{i=1}^{l'_L}d_i + 2 + \sum_{i=1}^{s_{t}+1}d_{l'_R-i}$ steps to reach its final position in $P$. Since $s_t, l'_L$ and $l'_R$ are all $\le p_{ij}$ we have that after
	\begin{align}
	\sum_{i=1}^{l'}d_i + 2 + \sum_{i=1}^{s_{t}+1}d_{p-i} \le 2\max(\alpha_i,\alpha_j) p_{ij} + O(1)
	\end{align}
	steps $P_L$ would be pure. Since, $P_L$ is arbitrary we see that $T'_i$ must be all 0's as well. \qed
\end{proof}

The $\mathsf{Swap}(T_i,T_{j};r)$ procedure is called for each consecutive pairs of subtrees during one pass. We call a subtree type-1 if it contains all 1's. Similarly we define a tree to be type-0 if it contains all 0's. Next we prove the following assertion: after one pass one of the following is true. 1) The rightmost subtree $T_d$ is of type-1. 2) All the subtrees to its left are type-0. This can be easily seen from the fact that the trees are arranged from left to right with decreasing size. After the first swap between $T_1$ and $T_2$ it is obvious that either $T_1$ is of type-0 or $T_2$  is of type-1. Let the assertion hold after completion of $\mathsf{Swap}(T_i,T_{i+1};r)$. Hence, either $T_i$ is of type-1 or all the subtrees before $T_i$ are of type-0. If the former is true then after completion of $\mathsf{Swap}(T_i,T_{i+1};r)$, $T_{i+1}$ would be of type-1 since $n_i \ge n_{i+1}$. If the latter holds then after the swap all subtrees to the left of $T_{i+1}$ will be of type-0. This proves the assertion. Hence during the next pass we can ignore $T_d$ \footnote{At this stage we could start sorting $T_d$ in parallel, however, since we have wait for $T_1$ to start sorting in parallel, this does not improve the number of stages.}. In the above analysis we see that the number of cycles depends on $p_{ij}$ which is a property of the input sequence. However we can easily make $\mathsf{Swap}(T_i,T_{i+1};r)$ oblivious by observing that $p^j_{i,i+1} \le n_{i+j}$, where $p^j_{i,i+1}$ is the number of pebbles exchanged between $T'_i$ and $T_{i+1}$ during the $j^{th}$ pass. Let us consider the swap between $T'_1$ and $T_2$. During the first pass $\le n_2$ pebbles are exchanged. But during the first pass at most $n_3$ pebbles were exchanged between $T'_2$ and $T_3$ hence at most this many pebble will be exchanged during the swap operation between $T'_1$ and $T_2$ during the second pass. We see that this argument can be generalized easily. 

The total number of cycles in the $j^{th}$ pass is

\begin{align}
c_j = \sum_{i=1}^{d-j}{(2n_i+\beta_ip^j_{i,i+1})}
\end{align}

\noindent  where $\beta_i = \max(\alpha_i,\alpha_{i+1})$. So the total number of cycles during phase-1 is
\begin{align}
S(n,d; \alpha_1,\ldots,\alpha_d) &= \sum_{i=1}^{d-1}{c_j} = \sum_{j=1}^{d-1}\sum_{i=1}^{d-j}{(2n_i+\beta_ip^j_{i,i+1})} \le 2dn + \sum_{i=1}^{d-1}{\beta_i\sum_{j=1}^{d-i}{p^j_{i,i+1}}} 
\end{align}

\noindent At this point we focus on the double sum. We want to show $\sum_{i=1}^{d-1}{\beta_i\sum_{j=1}^{d-i}{p^j_{i,i+1}}}  \le c\min(\Delta^2 n,n^2)$. Here $\Delta$ is the maximum degree of $T$ and $c$ is some constant.  The total number of pebbles transported to $T_i$ after $d-1$ passes is $\le$ to the total number of nodes in the trees right of $T_i$. Hence, $\sum_{j=1}^{d-i}{p^j_{i,i+1}} \le \sum_{j=i+1}^{d}{n_j}$. We have,
\begin{align}
\nonumber \sum_{i=1}^{d-1}{\beta_i\sum_{j=1}^{d-i}{p^j_{i,i+1}}}  &\le \sum_{i=1}^{d-1}{\beta_i\sum_{j=i+1}^{d}{n_j}} \le \left(\sum_{i=1}^{d-1}{\beta_i}\right)\left(\sum_{i=2}^{d}{n_i}\right) \\&\le \left(\sum_{i=1}^{d-1}{\beta_i}\right)(n/2) \le cn\min(\Delta^2,n)
\end{align}

\noindent The last inequality follows from the fact that $d, \beta_i \le \Delta$ for all $i$ and $\sum_{i=1}^{d-1}{\beta_i} \le n$.  Putting this upper bound of $S(n,d; \alpha_1,\ldots,\alpha_d) $ in Equation 1 we get a simplified recurrence:

\begin{align}
C(n) \le C(n/2) + O(\min(\Delta^2n,n^2))
\end{align}

\noindent  This shows that $\mathsf{OddEvenTreeSort}$ requires $O(\min{(\Delta^2n,n^2)})$ stages to correctly sort any input with respect to the MP ordering. We note that for the two extreme cases, 1) when $T=P_n$ and 2) $T = K_{1,n-1}$ the number stages needed by $\mathsf{OddEvenTreeSort}$ is optimal up to a constant factor. It remains to be seen if there exists an $O(\Delta n)$ round sorting network for trees

%
%
%

%
%
%


\end{document}